\newcommand{\longversion}[1]{#1}
\newcommand{\shortversion}[1]{}

\documentclass[a4paper,10pt]{article}

\usepackage{amssymb,amsmath}
\usepackage{color}
\usepackage{makecell}
\usepackage{multirow}
\usepackage{float}
\usepackage{graphicx}

\newenvironment{proof}{\noindent{\em Proof:}}{ \hfill $\square$\\ }

\usepackage{enumitem}

\setlist[itemize]{noitemsep, topsep=2pt}

\DeclareMathOperator*{\argmax}{arg\,max}

\usepackage[usenames,svgnames,dvipsnames]{xcolor}

\longversion{
\usepackage{fullpage}
\usepackage{charter,eulervm}
\usepackage[colorlinks=true,linkcolor=blue,citecolor=Mahogany]{hyperref}

}

\allowdisplaybreaks

\newdimen\prevdp
\def\leftlabel#1{\noalign{\prevdp=\prevdepth
   \kern-\prevdp\nointerlineskip\vbox to0pt{\vss\hbox{\ensuremath{#1}}}\kern\prevdp}}

\usepackage{url}

\usepackage{xspace}
\newcommand{\NP}{\ensuremath{\mathsf{NP}}\xspace}
\newcommand{\NPC}{\ensuremath{\mathsf{NP}}-complete\xspace}
\newcommand{\NPH}{\ensuremath{\mathsf{NP}}-hard\xspace}

\newcommand{\coNPH}{\ensuremath{\textsf{co-NP-hard}}\xspace}
\newcommand{\el}{\ensuremath{\ell}\xspace}
\newcommand{\suc}{\ensuremath{\succ}\xspace}

\newcommand{\Pb}{\ensuremath{P}\xspace}

\newcommand{\YES}{{\sc yes}\xspace}
\newcommand{\NO}{{\sc no}\xspace}

\newcommand{\PS}{{\sc Permutation Sum}\xspace}
\newcommand{\SM}{{\sc Distance Restricted Strong Manipulation}\xspace}
\newcommand{\DRSM}{{\sc DRSM}\xspace}
\newcommand{\DRWM}{{\sc DRWM}\xspace}
\newcommand{\LB}{{\sc Local Distance Constrained Bribery} \xspace}
\newcommand{\PM}{{\sc Distance Restricted Weak Manipulation}\xspace}

\newcommand{\MAN}{{\sc Manipulation}\xspace}
\newcommand{\XTC}{{\sc X3C}\xspace}
\newcommand{\XTCC}{$\overline{\text{\sc X3C}}$\xspace}

\renewcommand{\AA}{\ensuremath{\mathcal A}\xspace}
\newcommand{\BB}{\ensuremath{\mathcal B}\xspace}
\newcommand{\CC}{\ensuremath{\mathcal C}\xspace}
\newcommand{\DD}{\ensuremath{\mathcal D}\xspace}
\newcommand{\EE}{\ensuremath{\mathcal E}\xspace}

\newcommand{\GG}{\ensuremath{\mathcal G}\xspace}

\newcommand{\LL}{\ensuremath{\mathcal L}\xspace}
\newcommand{\MM}{\ensuremath{\mathcal M}\xspace}

\newcommand{\OO}{\ensuremath{\mathcal O}\xspace}
\newcommand{\PP}{\ensuremath{\mathcal P}\xspace}
\newcommand{\QQ}{\ensuremath{\mathcal Q}\xspace}
\newcommand{\RR}{\ensuremath{\mathcal R}\xspace}
\renewcommand{\SS}{\ensuremath{\mathcal S}\xspace}
\newcommand{\TT}{\ensuremath{\mathcal T}\xspace}
\newcommand{\UU}{\ensuremath{\mathcal U}\xspace}
\newcommand{\VV}{\ensuremath{\mathcal V}\xspace}
\newcommand{\WW}{\ensuremath{\mathcal W}\xspace}
\newcommand{\XX}{\ensuremath{\mathcal X}\xspace}

\newcommand{\NB}{\ensuremath{\mathbb N}\xspace}

\usepackage{nicefrac}

\newtheorem{observation}{\bf Observation}
\newtheorem{theorem}{\bf Theorem}
\newtheorem{lemma}{\bf Lemma}

\newtheorem{definition}{\bf Definition}

\newcommand{\eps}{\ensuremath{\varepsilon}\xspace}
\renewcommand{\epsilon}{\eps}

\newcommand{\ignore}[1]{}

\newcommand{\pr}{\ensuremath{\prime}}

\renewcommand{\leq}{\leqslant}
\renewcommand{\geq}{\geqslant}
\renewcommand{\ge}{\geqslant}
\renewcommand{\le}{\leqslant}

\usepackage{cleveref}

\crefname{theorem}{theorem}{\bf Theorem}
\crefname{observation}{observation}{\bf Observation}
\crefname{lemma}{lemma}{\bf Lemma}
\crefname{corollary}{corollary}{\bf Corollary}
\crefname{proposition}{proposition}{\bf Proposition}
\crefname{definition}{definition}{\bf Definition}
\crefname{claim}{claim}{\bf Claim}
\crefname{reductionrule}{reduction rule}{\bf Reduction rule}

\title{Distance Restricted Manipulation in Voting}

\author{Aditya Anand and Palash Dey\\\texttt{adityaanand1998g@iitkgp.ac.in,palash.dey@cse.iitkgp.ac.in}\\ Indian Institute of Technology Kharagpur}

\begin{document}

\maketitle
	
\begin{abstract}
	We introduce the notion of {\em Distance Restricted Manipulation}, where colluding manipulator(s) need to compute if there exist votes which make their preferred alternative win the election when their knowledge about the others' votes is a little inaccurate. We use the Kendall-Tau distance to model the manipulators' confidence in the non-manipulators' votes. To this end, we study this problem in two settings - one where the manipulators need to compute a manipulating vote that succeeds irrespective of perturbations in others' votes ({\em Distance Restricted Strong Manipulation}), and the second where the manipulators need to compute a manipulating vote that succeeds for at least one possible vote profile of the others ({\em Distance Restricted Weak Manipulation}). We show that {\em Distance Restricted Strong Manipulation} admits polynomial-time algorithms for every scoring rule, maximin, Bucklin, and simplified Bucklin voting rules for a single manipulator, and for the $k$-approval rule for any number of manipulators, but becomes intractable for the Copeland$^\alpha$ voting rule for every $\alpha\in[0,1]$ even for a single manipulator. In contrast, {\em Distance Restricted Weak Manipulation} is intractable for almost all the common voting rules, with the exception of the plurality rule. For a constant number of alternatives, we show that both the problems are polynomial-time solvable for every anonymous and efficient voting rule.
\end{abstract}

\begin{table*}[]
\centering

\resizebox{\textwidth}{!}{\begin{tabular}{|c|c|c|}\hline
 	
 Voting Rule & \SM & \PM\\\hline\hline

 Plurality & \makecell{$\mathbf{\Pb}$ for any number of manipulators   [\Cref{thm:kapp}]} & \makecell{$\mathbf{\Pb}$ for any number of manipulators[\Cref{thm:plurality_pm}]} \\\hline

 $k$-approval, $k \geq 2$ & \makecell{$\mathbf{\Pb}$ for any number of manipulators   [\Cref{thm:kapp}]} & \makecell{$\textbf{\NPC}$, even with $\delta = 2$ for each vote, \\for single manipulator[\Cref{thm:kapp_pm}]} \\\hline
 
 Scoring rules & \makecell{$\mathbf{\Pb}$ for single manipulator  [\Cref{thm:sc}]} &  \makecell{$\textbf{\NPC}$ for \textbf{Borda}, even with $\delta = 1$ \\ for each vote, for single manipulator[\Cref{thm:borda_pm}]} \\\hline
 
 Maximin & \makecell{$\mathbf{\Pb}$ for single manipulator [\Cref{thm:maximin}]\\\NPH for $\ge\!2$ manipulators\\\relax [\Cref{obs}, Faliszewski et al.~\cite{FaliszewskiHS08}~\cite{FaliszewskiHS10}]} &  \makecell{$\textbf{\NPC}$, even with $\delta = 3$ for each vote, \\for single manipulator[\Cref{thm:maximin_pm}]}\\\hline
 
 Copeland$^\alpha$ & \makecell{$\textbf{co-NP-Hard}$ even for\\ single manipulator [\Cref{thm:copeland}]} & \makecell{$\textbf{\NPC}$, even with $\delta = 3$ \\for each vote, \\ for single manipulator[\Cref{thm:copeland_pm}]} \\\hline
 
 Bucklin & \makecell{$\mathbf{\Pb}$ for single manipulator [\Cref{thm:buck}]} & \makecell{$\textbf{\NPC}$, even with $\delta = 1$ for each vote, \\for single manipulator[\Cref{thm:bucklin_pm}]} \\\hline 
 
 Simplified Bucklin & \makecell{$\mathbf{\Pb}$ for single manipulator  [\Cref{thm:sim_buck}]} & \makecell{$\textbf{\NPC}$, even with $\delta = 1$ for each vote, \\for single manipulator[\Cref{thm:sb_pm}]} \\\hline
 \end{tabular}
 }
 \caption{Summary of results for \SM and \PM. Our algorithms work even for the case when manipulators have different $\delta$ value for different voters. Our hardness results hold even when the manipulators have the same $\delta$ value for every voter. Results in bold are proved in this paper.}\label{tbl:summary}
\end{table*}

\section{Introduction}

Voting has served as a fundamental tool for aggregating preferences of a set of people over a set of alternatives for centuries. A typical voting system consists of a set of alternatives, a set of voters each having a linear order over the set of alternatives as her preference, and a voting rule which selects a set of alternatives as winners depending on the voters' preferences. However, classical results show that every reasonable voting system with at least $3$ alternatives can suffer from manipulation~\cite{gibbard1973manipulation,satterthwaite1975strategy} --- an agent may be able to make her more favored alternative win by misreporting her preference. Bartholdi et al. pioneered the idea of using computational intractability as a barrier to safeguard elections against manipulation~\cite{bartholdi1989computational,bartholdi1991single}. Indeed, if we have $m$ alternatives and even if the manipulator exactly knows the preferences of all other voters, na\"{i}vely going over all $(m!-1)$ possible preferences and reporting the one that results in the best outcome for the manipulator is not feasible for any computationally bounded manipulator.

Although the idea of Bartholdi et al. was to use computational intractability as a barrier against manipulation, the computational problem of manipulation admits efficient algorithm for most of the commonly used voting rules such as the scoring rules, maximin, Copeland, etc. with the prominent exception of the single transferable vote (STV) voting rule. Even for voting rules (STV for example) for which the computational barrier exists against manipulation, it seems that the barrier, in reality, maybe substantially weak due to the existence of heuristics which work well in practice~\cite[and references there in]{FaliszewskiP10,FriedgutKKN11,MosselR15}.

\paragraph*{Motivation:} The computational problem of manipulation has mostly been studied in what is called the complete information setting --- the manipulator knows exactly, the preferences of all other voters. Although this setting may be the best possible to prove intractability results (if one proves that manipulation is intractable even if the manipulator exactly knows the preferences of all other voters, then the manipulator's job can only be harder if she does not know some part of the others' preferences), it is hardly practical. Indeed, most applications of voting in AI - voting over a social network for example, involve a large number of voters where the complete information setting is far from reality. This motivates us to study the classical manipulation problem in an incomplete information setting. In our model, for every voter $v$, the manipulator has a believed preference $\suc_v$ and an integer $\delta_v$ denoting the worst-case Kendall-Tau (number of pairs which are ranked differently) distance by which the true preference of the voter $v$ can deviate from $\suc_v$ ; low (high respectively) value of $\delta_v$ corresponds to the manipulator having high (low respectively) confidence on her belief about voter $v$'s true preference. Indeed, in many real-world election scenarios, the manipulator can form a belief about a voter's preference based on that voter's historical data and other activities. However, due to various activities that may have happened since the last election or simply because of the inherent uncertainty in human nature, the voter's preference may have slightly changed (quantified as $\delta_v$).

\subsection{Contribution}

Our basic problem extends the problem of (coalitional) manipulation - the input is a set \AA of $m$ alternatives, a set \VV of $n$ voters, associated with each of them a believed preference, a distinguished alternative $c$, and a number \el of manipulators. In the \SM (\DRSM) problem, we need to compute if there exists a preference profile for the manipulators which makes $c$ win the election irrespective of any deviation of every other voter $v$ from her believed preference $\suc_v$ by at most $\delta_v$ under Kendall-Tau distance. In the \PM (\DRWM) problem, we need to compute if there exists at least one preference profile for the manipulators which makes $c$ win the election in at least one profile of the non-manipulators, where the deviation of every non-manipulating voter $v$ from her believed preference $\suc_v$ is at most $\delta_v$ under Kendall-Tau distance.

These two settings are motivated as follows: from the manipulators' perspective, \DRSM looks for guaranteed success irrespective of small perturbations, while \DRWM examines the possibility of manipulation of a given profile. Further, if we consider a more general problem in which we ask if manipulators can successfully manipulate a fixed threshold fraction of profiles of non-manipulators which meet the Kendall-Tau distance constraints, hardness results for both \DRWM(which is computationally hard for most common voting rules) and \DRSM transfer to this natural setting.

We summarize our complexity-theoretic results in \Cref{tbl:summary}. Other than that, we show that, for a constant number of alternatives, both problems are polynomial-time solvable for every anonymous and efficient voting rule for any number of manipulators~[\Cref{thm:poly}].

\subsection{Related Work}

Initiated by Bartholdi et al. the study of the computational complexity of manipulation has been one of the key research areas in computational social choice~\cite{bartholdi1989computational,bartholdi1991single}. Conitzer et al. showed that, for weighted elections, the coalition manipulation problem which is manipulation by a coalition of voters, is \NPC even when we have a small constant number of alternatives for most of the commonly used voting rules~\cite{ConitzerSL07}. Faliszewski and Procaccia exhibited evidence that the computational problem of manipulation may not be computationally challenging on average~\cite{FaliszewskiP10}. Mossel and R{\'{a}}cz and Friedgut et al. showed that, for a uniformly random preference profile, reporting a random preference results in a successful manipulation with high probability ($1$ over some polynomial in the number of voters and the number of alternatives) for any reasonable voting rule~\cite{MosselR15,FriedgutKKN11}. We refer to \cite{ConitzerW16} for an excellent overview of the computational problem of manipulation. Manipulation comes under a more general class of problems known as election control problems. Election control refers to the phenomenon of influencing the outcome of an election through various means. Other than manipulation, prominent examples of election control problems include bribery, voter deletion, alternative deletion, voter partition, alternative partition, etc. We refer to~\cite{FaliszewskiR16} for a comprehensive survey of various kinds of election control problems.

The effect of limiting manipulators' access to other voters' preference profile on the computational complexity of manipulation has been studied before. Conitzer et al. defined the dominating manipulation problem in a bid to model the manipulator's limited information and showed that the commonly used voting rules, except plurality and veto, are resistant to this kind of manipulation~\cite{ConitzerWX11}. Dey et al. captured the manipulator's approach to risk into the concept of weak, strong, and opportunistic manipulation and showed that the weak as well as opportunistic manipulations are intractable for all the commonly used voting rules except plurality and veto whereas the strong manipulation problem admits a polynomial-time algorithm for most of the common voting rules~\cite{DeyMN18}. Both the above papers model the manipulator's limited information as a partial preference profile (a partial preference for each voter); for every other voter, the manipulator knows the ordering of some pairs of alternatives for sure but does not have any clue about the remaining pairs. On the other hand, our model of the manipulator's limited information cannot be modelled as partial preferences over voters; in our model, intuitively speaking, the manipulator's lack of information is distributed over all pairs of alternatives.

 By adopting a quantitative deviation measure, we capture small changes in voters' mindset from manipulators' beliefs, as opposed to the models of partial information which assume that the manipulators know exactly the rankings of certain pairs for each voter. This is especially relevant, say, if the past voting patterns of a voter are available: one would expect that the voter would only change his preference by a small distance. Further, it maybe possible to learn a good bound on this distance for each individual, by observing changes in past voting.

The effect of incomplete information has been studied in other (different from computational) settings. A related line of work studies voting equilibria when voters have access to only partial information and vote strategically. Meir et al.~\cite{MeirLR14} and Meir~\cite{Meir15} proposed a theory of voting equilibria based on local dominance and showed convergence to such equilibria under local-dominance based dynamics for the plurality voting rule. Lev et al. presented a theory of ordinally-dominated strategies and showed its effectiveness to justify known voting heuristics as bounded rational~\cite{LevMOP19}. Reijngoud and Endriss~\cite{ReijngoudE12} and Endriss et al.~\cite{EndrissOPR16} studied the manipulation problem for popular voting rules when the voters have limited access to others' votes through an opinion poll. Further, they studied voter response to iterated poll information in an iterative voting setting. Slinko and White~\cite{SlinkoW14} studied manipulation in the setting with more than one manipulator, each with same true preference, where each manipulator contemplates casting the same strategic vote as the others. In such a situation, mis-coordinated strategic voting by the manipulators may lead to a worse outcome for all of them. They showed that for every onto, non-dictatorial voting rule there are circumstances where a manipulator can cast a strategic vote safely.
\section{Preliminaries}\label{sec:prelim}

Let us denote the set $\{1, 2, \ldots, \el\}$ by $[\el]$ for any positive integer $\el$. Let $\AA$ be a set of $m$ alternatives and $\VV$ a set of $n$ voters. If not mentioned otherwise, we denote the number of alternatives by $m$ and the number of voters by $n$. Every voter $v_i$ has a preference or vote $\suc_i$ which is a complete order over \AA. We denote the set of complete orders over \AA by $\LL(\AA)$. We call a tuple of $n$ preferences $(\suc_1, \suc_2, \cdots, \suc_n)\in\LL(\AA)^n$ an $n$-voter preference profile. An election is defined as a set of alternatives together with a voting profile. Let $\uplus$ denote the disjoint union of sets. A map $r:\uplus_{n,|\AA|\in\mathbb{N}^+}\mathcal{L(\AA)}^n\longrightarrow 2^\mathcal{\AA}\setminus\{\emptyset\}$ is called a \emph{voting rule}. A voting rule $r$ is called {\em efficient} if the winners under $r$ can be computed in polynomial time. A voting rule is called anonymous if the set of winners does not depend on the names of the voters. For a voting rule $r$ and a preference profile $\succ = (\succ_1, \dots, \succ_n)$, we say an alternative $x$ wins uniquely if $r(\succ) = \{x\}$ and $x$ co-wins if $x\in r(\suc)$. For a vote $\suc\in\LL(\AA)$ and two alternatives $x, y\in\AA$, we say that $x$ is placed before $y$ in \suc if $x\suc y$; otherwise we say $x$ is placed after $y$ in~\suc. An alternative is said to be at the $i^{th}$ position from the top/left (bottom/right) if there are exactly $(i-1)$ alternatives before (after) it. For any two alternatives $x, y\in\AA$ with $x\ne y$ in an election $\EE=(\AA,\PP)$, let us define the margin $\DD_\PP(x, y)$ of $x$ from $y$ to be $|\{ i: x \suc_i y \}| - |\{ i: y \suc_i x \}|$. Examples of some common voting rules are as follows.

{\bf Positional scoring rules:} A collection $(\overrightarrow{s_m})_{m\in\NB^+}$ of $m$-dimensional vectors $\overrightarrow{s_m}=\left(\alpha_1,\alpha_2,\dots,\alpha_m\right)\in\mathbb{N}^m$ 
 with $\alpha_1\ge\alpha_2\ge\dots\ge\alpha_m$ and $\alpha_1>\alpha_m$ for every $m\in \mathbb{N}^+$ naturally defines a voting rule --- an alternative gets score $\alpha_i$ from a vote if it is placed at the $i^{th}$ position from the top, and the score of an alternative is the sum of the scores it receives from all the votes. 
 The winners are the alternatives with maximum score. If $\alpha_i$ is $1$ for $i\in [k]$ and $0$ otherwise, then we get the $k$-approval voting rule. If $\alpha_i=m-i$, then we get Borda rule.

\textbf{Copeland$^{\alpha}$:} Given $\alpha\in[0,1]$, the Copeland$^{\alpha}$ score of an alternative $x$ in an election $\EE=(\AA,\PP)$ is $|\{y \in \AA \setminus x:\DD_\PP(x,y)>0\}|+\alpha|\{y \in \AA \setminus x:\DD_\PP(x,y)=0\}|$. The winners are the alternatives with maximum Copeland$^{\alpha}$ score. 

{\bf Maximin:} The maximin score of an alternative $x$ in an election $\EE=(\AA,\PP)$ is $\min_{y \in \AA \setminus x} \DD_\PP(x,y)$. The winners are the alternatives with maximum score.

{\bf Bucklin and simplified Bucklin:} Let $\ell$ be the minimum integer such that there exists at least one alternative $x\in\AA$ who more than half of the voters place in their top $\ell$ positions. Then the Bucklin winners are the alternatives who are placed the most number of times within the top $\el$ positions of the votes. The simplified Bucklin winners are the alternatives who appear within the top \el positions in a majority of the preferences.
 

The {\em Kendall-Tau} distance between a pair of preferences $\suc,\suc^\pr\in\LL(\AA)$, denoted by $d_{KT}(\suc,\suc^\pr)$, is the number of pairs of alternatives where \suc and $\suc^\pr$ differ; that is $d_{KT}(\suc,\suc^\pr)=|\{(a,b)\in\AA\times\AA: a\ne b, a\suc b, b\suc^\pr a\}|$. Alternatively, the Kendall-Tau distance between two preferences is the minimum number of adjacent swaps needed to convert a preference into another. In this draft, by swaps, we mean only adjacent swaps.
%
%
We now define our problems formally.

\begin{definition}[\\\DRSM]\label{def:gsb}
 Given a set \AA of $m$ alternatives, an $n$-voter profile $\PP=(\suc_i)_{i\in[n]}\in\LL(\AA)^n$ over \AA, a distinguished alternative $c\in\AA$, a tuple $(\delta_i)_{i\in[n]}$ of non-negative integers, and the number \el (a positive integer) of manipulators, compute if there exists a profile $(\suc_{n+1}^\pr,\ldots,\suc_{n+\el}^\pr)\in\LL(\AA)^\el$ such that, we have $c\in r((\suc_i^\pr)_{i\in[n+\el]})$ for every $(\suc_i^\pr)_{i\in[n]}\in\LL(\AA)^n$ with $d_{KT}(\suc_i,\suc_i^\pr)\le\delta_i$ for every $i\in[n]$.
\end{definition}

\begin{definition}[\\\DRWM]\label{def:gsb}
 Given a set \AA of $m$ alternatives, an $n$-voter profile $\PP=(\suc_i)_{i\in[n]}\in\LL(\AA)^n$ over \AA, a distinguished alternative $c\in\AA$, a tuple $(\delta_i)_{i\in[n]}$ of non-negative integers, and the number \el (a positive integer) of manipulators, compute if there exists a profile $(\suc_{n+1}^\pr,\ldots,\suc_{n+\el}^\pr)\in\LL(\AA)^\el$ such that, we have $c\in r((\suc_i^\pr)_{i\in[n+\el]})$ for some $(\suc_i^\pr)_{i\in[n]}\in\LL(\AA)^n$ with $d_{KT}(\suc_i,\suc_i^\pr)\le\delta_i$ for every $i\in[n]$.
\end{definition}

We denote an arbitrary instance of above problems by $(\AA,\PP,c,(\delta_i)_{i\in[n]},\el)$. The above definition requires $c$ to be a co-winner. One can similarly pose the problem in the unique winner setting. We remark that all our results, both algorithmic and hardness, easily extend to the unique winner setting. For ease of exposition, we restrict ourselves to the co-winner setting only in this short version. For the classical \MAN problem, we have $\delta_i=0$ for every $i\in[n]$. Hence, we have the following observation.

\begin{observation}\label{obs}
	If \MAN is \NPH for a voting rule $r$ with $\el$ manipulators, then both \DRSM and \DRWM are \NPH for $r$ with $\el$ manipulators.
\end{observation}



\graphicspath{ {./experiments/} }
\section{Results}

We present our theoretical results in this section. We begin by presenting our algorithm for the scoring rules.

\subsection{Results for \DRSM}

\begin{theorem} \label{thm:sc}
 There is a polynomial-time algorithm for the \DRSM problem for every scoring rule if we have only one manipulator.
\end{theorem}

\begin{proof}
We make use of a greedy construction similar to the one used in~\cite{bartholdi1989computational}. On a high level, we greedily place alternatives while they can be {\em``safely"} placed. The safety of a position for an alternative $a$ is checked by creating the `worst possible profile' for the candidate $c$ w.r.t $a$.

 Let $\alpha=(\alpha_i)_{i\in[m]}$ be an arbitrary scoring rule and $(\AA,\PP,c,(\delta_i)_{i\in[n]},\el=1)$ be an arbitrary instance of \DRSM for $\alpha$. We iteratively try to construct a manipulator's preference $\suc_M$ which results in a successful manipulation. Without loss of generality, we place the alternative $c$ at the first position of $\suc_M$ in the first iteration. Iteratively, suppose we have already placed alternatives at every position in $\{1,\ldots,t-1\}$ for some $2\le t\le m$ and we next wish to place an alternative at the $t$-th position. Let $\AA_{t-1}$ be the set of alternatives which are placed within the first $t-1$ positions; we obviously have $c\in\AA_{t-1}$. We now check if there exists an alternative $a\in\AA\setminus\AA_{t-1}$ which can be placed at the $i$-th position {\em``safely"}, then we place the alternative $a$ at the $i$-th position and go to the next iteration; otherwise we output that the instance is a \NO instance. We say that the position $t$ is {\em``safe"} for an alternative $a\in\AA\setminus\AA_{t-1}$ if there does not exist any $n$-voter profile \QQ such that (i) the Kendall Tau distance between the $j$-th preferences of \PP and \QQ is at most $\delta_j$ for every $j\in[n]$ and (ii) the score of the alternative $a$ is more than the score of the alternative $c$ in the profile $(\QQ,\suc_M)$ where $\suc_M$ is any preference which places the alternatives $c$ and $a$ at positions $1$ and $t$ respectively. We next describe how to check, in polynomial-time, whether a position $t\in\{2,\ldots,m\}$ is safe for an alternative $a$.
 
 Before proceeding further, let us define some notation. For a preference $\suc\in\LL(\AA)$ and an alternative $a\in\AA$, let $\text{rank}(\suc,a)$ be the position of the alternative $a$ in the preference order $\suc$. We define $RS(\suc,a,k)$ to be the preference $\suc^\pr$ obtained by shifting the alternative $a$ to the right by $\min(k,m-\text{rank}(\suc,a))$ positions. Similarly, we define $LS(\suc,a,k)$ to be the preference $\suc^\pr$ obtained by shifting the alternative $a$ to the left by $\min(k,\text{rank}(\suc,a)-1)$ positions. We use $S(\suc,a)$ and $S(\RR,a)$ to denote the score of the alternative $a$ in a preference \suc and a preference profile \RR respectively. From the given preference profile $\PP=(\suc_i)_{i\in[n]}$, we construct another preference profile $\QQ^a=(\suc_i^\pr)_{i\in[n]}$ as follows. For every $i\in[n]$, let $j_i\in\{0,1,\ldots,\delta_i\}$ be the integer by which degrading the position of the alternative $c$ and followed by improving the position of the alternative $a$ in $\suc_i$, is worst possible for $c$ with respect to $a$ in $\suc_i$. Formally, for $j\in\{0,1,\ldots,\delta_i\}$, let $\Delta(\suc_i,c,j,a,\delta_i-j)$ be the decrease of the score of the alternative $c$ plus the increase in the score of $a$ if we degrade the position of $c$ in the preference $\suc_i$ by $j$ and then we improve the position of $a$ by $\delta_i-j$; that is $\Delta(\suc_i,c,j,a,\delta_i-j)=S(\suc_i,c)-S(LS(RS(\suc_i,c,j),a,\delta_i-j),c)+S(LS(RS(\suc_i,c,j),a,\delta_i-j),a)-S(\suc_i,a)$. Then we choose $j_i\in\argmax_{j\in\{0,1,\ldots,\delta_i\}} \Delta(\suc_i,c,j,a,\delta_i-j)$. We define $\suc_i^\pr = LS(RS(\suc_i,c,j_i),a,\delta_i-j_i)$. For an alternative $a\in\AA\setminus\AA_{t-1}$ and a position $t$, we say that the position $t$ (in the manipulator's vote) is safe for the alternative $a$ if $S(\QQ^a,c)+\alpha_1 \ge S(\QQ^a,a)+\alpha_t$. This concludes the description of our algorithm. Clearly, our algorithm runs in polynomial-time. We next prove its correctness.
 
 Suppose the algorithm outputs that the input instance is a \YES instance. Then we claim that the manipulator's preference $\suc_M$ is a successful manipulation. Suppose not, then there exists an alternative $a\in\AA\setminus\{c\}$ and an $n$-voters preference profile $\RR^a$ such that (i) the Kendall Tau distance between the $j$-th preferences of \PP and $\RR^a$ is at most $\delta_j$ for every $j\in[n]$ and (ii) the score of the alternative $a$ is more than the score of the alternative $c$ in the profile $(\RR^a,\suc_M)$. Suppose the position of the alternative $a$ in $\suc_M$ is $j_a\in\{2,\ldots,m\}$. Then, from the design of the algorithm, it follows that $S(\QQ^a,c)+\alpha_1 \ge S(\QQ^a,a)+\alpha_{j_a}$ where $\QQ^a$ is the profile considered by the algorithm in the $j_a$-th iteration (when the alternative $a$ was placed at the $j_a$-th position) for checking safety of $a$ at position $j_a$. From the construction of $\QQ^a$, it follows that $S(\QQ^a,a)-S(\QQ^a,c) \ge S(\RR^a,a)-S(\RR^a,c)$ and thus we have the following 
 $$ S(\RR^a,a)-S(\RR^a,c) \le S(\QQ^a,a)-S(\QQ^a,c) \le \alpha_1-\alpha_{j_a} $$
 which implies that $S(\RR^a,c)+\alpha_1 \ge S(\RR^a,a)+\alpha_{j_a}$. However, this contradicts our assumption that the score of the alternative $a$ is more than the score of the alternative $c$ in the profile $(\RR^a,\suc_M)$. Hence the instance was indeed a \YES instance. Now suppose that the algorithm outputs that the input instance is a \NO instance. Then, there exists an integer $t\in\{2,\ldots,m\}$ such that the algorithm does not find any alternative in the $t$-th iteration to be placed at position $t$ safely. We observe that, if a position $k$ is unsafe for an alternative $x\in\AA\setminus\{c\}$, then the position $k-1$ is also unsafe for $x$. Then we have $m-t+1$ alternatives, namely the alternatives in the set $\AA\setminus\AA_{t-1}$, who must appear within the rightmost $m-t$ positions of any manipulator's preference $\suc_M^\pr$ if $\suc_M^\pr$ were to result in a successful manipulation, which is, by pigeonhole principle, impossible. Hence the input instance was indeed a \NO instance\longversion{ and thus the algorithm is correct.}
\end{proof}

 We use the same greedy strategy for the maximin rule. This time, however, the maximin scores of the $c$ and the alternative $a$ in consideration must be examined across all valid profiles for checking safety while placing $a$ in the manipulator's preference. We resolve this by first guessing the alternative against whom $c$ has the worst pairwise election.

\begin{theorem}\label{thm:maximin}
 There exists a polynomial-time algorithm for the \DRSM problem for the maximin voting rule if we have only one manipulator.
\end{theorem}

\longversion{

\begin{proof}
 Let $(\AA,\PP,c,(\delta_i)_{i\in[n]},\el=1)$ be an arbitrary instance of \DRSM for the maximin voting rule. On a high level, our algorithm for the maximin voting rule is similar to our algorithm for scoring rules: we put $c$ at the first position of the manipulator's vote in the first iteration, and then iteratively, in the $t$-th iteration, if $\AA_{t-1}$ is the set of alternatives within the first $t-1$ positions, we place an alternative $a\in\AA\setminus\AA_{t-1}$ if it is safe to do so; that is, given the partial preference of the manipulator constructed so far, placing the alternative $a$ at the $t$-th position does not make the maximin score of $a$ become more than the maximin score of the alternative $c$ for any $n$-voters preference profile \QQ where the Kendall Tau distance between the $i$-th preferences of \PP and \QQ is not more than $\delta_i$. The only thing that changes here from the algorithm in \Cref{thm:sc} is the algorithm for checking safety which we explain below.
 
 To introduce an algorithm for checking safety, we begin by assuming that the given alternative $a$ cannot be placed safely at the $t$-th position given the alternatives at the first $t-1$ positions; that is, there exists an $n$-voters preference profile \QQ with the properties stated above so that $a$'s maximin score is higher than that of $c$. We first guess an alternative $b\in\AA\setminus\{c\}$ (the alternative $b$ can be the alternative $a$ itself) such that the maximin score of $c$ in \QQ is $\DD_\QQ(c,b)$. From the given preference profile $\PP=(\suc_i)_{i\in[n]}$, we construct another preference profile $\QQ^a_b=(\suc_i^\pr)_{i\in[n]}$  so that  (i) the Kendall Tau distance between the $i$-th preferences of \PP and $\QQ^a_b$ is not more than $\delta_i$, and (ii) the difference between the maximin score of $a$ and $c$ is the maximum possible (the maximin score of $a$ being higher). For an $i\in[n]$, the preference $\suc_i$ can be one of the following types:
 
 {\em Case I -- it is possible to place the alternative $c$ on the right of the alternative $b$ by swapping at most $\delta_i$ pairs of alternatives:} Let $j_i$ be the minimum number of swaps needed in $\suc_i$ to place the alternative $c$ on the right of the alternative $b$; $j_i$ is $0$ if $c$ already appears on the right of $b$. Then we define $\suc_i^\pr = LS(RS(\suc_i,c,j_i),a,\delta_i-j_i)$; that is, we first shift $c$ right to place it immediately after $b$ and then shift $a$ left as much as we can.
 
 {\em Case II -- it is not possible to place the alternative $c$ on the right of the alternative $b$ by swapping at most $\delta_i$ pairs of alternatives:} We define $\suc_i^\pr = LS(\suc_i,a,\delta_i)$; that is, we shift $a$ left as much as we can.
 
 This finishes the description of the preference profile $\QQ^a_b$. Let $\suc_M^\pr$ be any arbitrary completion of the partially constructed preference of the manipulator. We declare the alternative $a$ to be safe at the $t$-th iteration if, for every alternative $b\in\AA\setminus\{c\}$, the alternative $c$ co-wins in the preference profile $(\QQ^a_b,\suc_M^\pr)$. This concludes the description of our algorithm. Clearly our algorithm runs in polynomial-time. We next prove its correctness.
 
 Suppose the algorithm outputs that the input instance is a \YES instance. Then we claim that the manipulator's preference $\suc_M$ is a successful manipulation. Suppose not, then there exists an alternative $a\in\AA\setminus\{c\}$ and an $n$-voters preference profile $\RR^a$ such that (i) the Kendall Tau distance between the $j$-th preferences of \PP and $\RR^a$ is at most $\delta_j$ for every $j\in[n]$ and (ii) the score of the alternative $a$ is more than the score of the alternative $c$ in the profile $(\RR^a,\suc_M)$. Suppose the position of the alternative $a$ in $\suc_M$ be $j_a\in\{2,\ldots,m\}$ and $b\in\AA\setminus\{c\}$ be an alternative such that the maximin score of the alternative $c$ is $\DD_{(\RR^a,\suc_M)}(c,b)$. Then, from the design of the algorithm it follows that the maximin score of $a$ is more than the maximin score of $c$ in the preference profile $(\QQ^a_b,\suc_M)$. This contradicts our assumption that the algorithm declared that placing the alternative $a$ at the $j_a$-th position in the $j_a$-th iteration was safe (since the algorithm must have placed the alternative $a$ at the $j_a$-th position in the $j_a$-th iteration). Hence, the input instance is indeed a \YES instance. Now suppose that the algorithm outputs that the input instance is a \NO instance. Then, there exists an integer $t\in\{2,\ldots,m\}$ such that the algorithm finds that it is unsafe for every alternative $a\in\AA\setminus\AA_{t-1}$ to appear before every alternative in $\AA\setminus(\AA_{t-1}\cup\{a\})$. However, in every possible manipulator's preference $\suc_M$, there exists an alternative $a\in\AA\setminus\AA_{t-1}$ which appears before every alternative in $\AA\setminus(\AA_{t-1}\cup\{a\})$. Hence the input instance is indeed a \NO instance and thus the algorithm is correct.
\end{proof}

}

For the $k$-approval voting rule, we are able to reduce the \DRSM problem with any number of manipulators to an equivalent maximum flow problem thereby obtaining a polynomial-time algorithm.

\begin{theorem}\label{thm:kapp}
There exists a polynomial-time algorithm for the \DRSM problem for the $k$-approval voting rule for any number of manipulators and any $k$.
\end{theorem}

\begin{proof}
 Let $(\AA,\PP,c,(\delta_i)_{i\in[n]},\el)$ be an arbitrary instance of \DRSM for the $k$-approval voting rule. We may assume without loss of generality that the alternative $c$ is placed at the first position in every preference of the manipulators. For every alternative $a\in\AA\setminus\{c\}$, we compute the maximum number $\lambda_a$ of manipulators' preferences where the alternative $a$ can appear within the first $k$ positions in any manipulators' preference profile which results in successful manipulation. Each preference $\suc_i, i\in[n]$ belongs to exactly one of the following types: (i) simultaneously $c$ can be placed outside of the first $k$ positions and $a$ can be placed within the first $k$ positions without changing order of more than $\delta_i$ pairs of alternatives in $\suc_i$ (ii) $c$ can be placed outside of the first $k$ positions without changing order of more than $\delta_i$ pairs of alternatives and $a$ can not be placed within the first $k$ positions without changing order of more than $\delta_i$ pairs of alternatives in $\suc_i$ (iii) $c$ can not be placed outside of the first $k$ positions without changing order of more than $\delta_i$ pairs of alternatives and $a$ can be placed within the first $k$ positions without changing order of more than $\delta_i$ pairs of alternatives in $\suc_i$ (iv) either $c$ can be placed outside of the first $k$ positions or $a$ can be placed within the first $k$ positions (but not both) without changing order of more than $\delta_i$ pairs of alternatives in $\suc_i$ (v) both $c$ can not be placed outside of the first $k$ positions without changing order of more than $\delta_i$ pairs of alternatives and $a$ can not be placed within the first $k$ positions without changing order of more than $\delta_i$ pairs of alternatives in $\suc_i$. Let the number of such preferences respectively be $n_1, n_2, n_3, n_4,$ and $n_5$. Let $S(\PP,x)$ be the $k$-approval score of any alternative $x\in\AA$ in the profile \PP. We define $\lambda_a=(\el + n_3 + n_5)-(n_1 + n_3) = \el + n_5 - n_1$; that is, loosely speaking, the worst profile for $c$ with respect to the alternative $a$ derived from \PP is to push $c$ outside first $k$, if possible, and then push $a$ within first $k$, if possible. If, for any alternative $a\in\AA\setminus\{c\}$, we have $\lambda_a<0$, then the algorithm outputs \NO.

 We now construct the following flow network $\GG=(\VV,\EE,s,t,c:E\longrightarrow\NB_{\ge1})$.
 \begin{align*}
 	\VV &= \{s,t\} \cup \{u_i: i\in[\el]\} \cup \{v_a | a \in \AA\setminus\{c\}\}\\
 	\EE &= \{(s,u_i):i\in[\el]\} \\
 	&\cup \{(u_i,v_a):i\in[\el],a\in\AA\setminus\{c\}\} \\
 	&\cup \{(v_a,t): a\in\AA\setminus\{c\}\}
 \end{align*}
 
 We now describe the capacities of the edges. The capacity of each outgoing edge from $s$ is $k-1$. For $a\in\AA\setminus\{c\}$, the capacity of the edge $(v_a,t)$ is $\lambda_a$. The capacity of every other edge is $1$. The algorithm outputs that the instance is a \YES instance if and only if there is an $s-t$ flow in \GG of value $\el(k-1)$. This finishes the description of the algorithm. We now prove its correctness.
 
 Suppose the algorithm outputs that the instance is a \YES instance. Then \GG has an $s-t$ flow $f$ of value $\el(k-1)$. We may assume without loss of generality that every edge carries an integral flow in $f$ since the capacity of every edge is some integer and $f$ is a maximum flow. We now construct a preference profile $\PP_M$ of \el manipulators. The manipulator $i$ places an alternative $a\in\AA\setminus\{c\}$ within first $k$ positions in $\PP_M$ if and only if the edge $(u_i,v_a)$ carries $1$ unit of flow under $f$. Since every manipulator places $c$ at the first position of her preference, and the incoming flow at vertex $u_i, i\in[\el]$ is $k-1$ under $f$, we have described which $k$ alternatives appear within the first $k$ positions of each manipulator's preference in $\PP_M$. We claim that the preference profile $\PP_M$ results in a successful manipulation. Indeed, for any $n$-voters profile \QQ where the Kendall Tau distance between the $i$-th preferences of \PP and \QQ is at most $\delta_i$, we have $S((\QQ,\PP_M),c)-S(\QQ,a)\ge\lambda_a$ for every alternative $a\in\AA\setminus\{c\}$. Since every alternative $a\in\AA\setminus\{c\}$ appears within the first $k$ positions at most $\lambda_a$ times in $\PP_M$, it follows that $\PP_M$ indeed results in successful manipulation. On the other hand, if the algorithm outputs \NO, then one of the following two cases happen. In the first case there exists an alternative $a\in\AA\setminus\{c\}$ such that $\lambda_a<0$. Consider the preference profile $\QQ^a$ obtained from \PP by simultaneously moving $c$ outside first $k$ positions in every preference of type (i), (ii), and (iv) and $a$ within first $k$ positions in every preference of type (i) and (iii). We observe that the $k$-approval score of $a$ in $\QQ^a$ is more than the $k$-approval score of $c$ from $\QQ^a$ plus \el and thus the instance is indeed a \NO instance. In the second case, suppose the algorithm outputs \NO because the maximum flow of \GG is strictly less than $\el(k-1)$. In this case too, for any possible preference profile $\PP_M$ of the manipulators, there exists an alternative $a\in\AA\setminus\{c\}$ which appears within first $k$ positions strictly more than $\lambda_a$ times. Then, in the profile $(\QQ^a,\PP_M)$, the $k$-approval score of the alternative $a$ is strictly more than the $k$-approval score of the alternative $c$. Hence the instance is indeed a \NO instance.
\end{proof}

In our algorithm for the simplified Bucklin rule, we check the safety of placing an alternative by counting certain types of votes.
\begin{theorem}\label{thm:sim_buck}
There exists a polynomial-time algorithm for the \DRSM problem for the simplified Bucklin voting rule if we have only one manipulator.
\end{theorem}

\longversion{
\begin{proof}
Let $(\AA,\PP,c,(\delta_i)_{i\in[n]},\el=1)$ be an arbitrary instance of \DRSM for the simplified Bucklin voting rule. On a high level, our algorithm for the simplified Bucklin voting rule is similar to our algorithm for scoring rules and the maximin voting rule. The only difference being, given a position $t\in\{2,\ldots,m\}$ and an alternative $a\in\AA\setminus\{c\}$, how do we decide if placing the alternative $a$ at position $t$ in the manipulator's vote is safe. We describe this below and skip repeating the other parts since they are exactly similar to the algorithms for the scoring rules and the maximin voting rule.

Let us denote by $S_k(\RR,a)$ the $k$-approval score of alternative $a$ in an $n$-voters  profile $\RR$. We observe that the alternative $c$ wins in \RR under the simplified Bucklin rule if and only if, for every other alternative $a\in\AA\setminus\{c\}$ and for every $k\in \{1,2,\ldots m\}$, $S_k(\RR,a) > \frac{n}{2} \implies S_k(\RR,c) > \frac{n}{2}$.  For any $k$ let (i) $n_1$ be the number of preferences $\suc_i\in\PP$ where simultaneously $c$ can be placed outside the first $k$ positions and $a$ can be placed within the first $k$ positions by swapping at most $\delta_i$ pairs of alternatives (call these preference type (i)), (ii) $n_2$ the number of preferences $\suc_i\in\PP$ where either $c$ can be placed outside the first $k$ positions or $a$ can be placed within the first $k$ positions by swapping at most $\delta_i$ pairs of alternatives but not both can be done (call these preference type (ii)), (iii) $n_3$ the the number of preferences $\suc_i\in\PP$ where $c$ can be placed outside the first $k$ positions and $a$ can not be placed within the first $k$ positions by swapping at most $\delta_i$ pairs of alternatives (call these preference type (iii)), (iv) $n_4$ the number of preferences $\suc_i\in\PP$ where $c$ can not be placed outside the first $k$ positions but $a$ can be placed within the first $k$ positions by swapping at most $\delta_i$ pairs of alternatives (call these preference type (iv)), (v) $n_5$ the number of preferences $\suc_i\in\PP$ where neither $c$ can be placed outside the first $k$ positions nor $a$ can be placed within the first $k$ positions by swapping at most $\delta_i$ pairs of alternatives (call these preference type (v)). We declare that position $t$ in the manipulator's preference is not safe for an alternative $a\in\AA\setminus\{c\}$ if there exists a position $k\in\{2,\ldots,m\}$ such that there exists an $n$-voters preference profile \QQ where (i) the Kendall Tau distance between the $i$-th preferences of \PP and \QQ is at most $\delta_i$ for every $i\in[n]$ and (ii) $c$ appears within the first $k$ positions in at most $(\lceil\frac{n}{2}\rceil-1)$ preferences in \QQ (observe that, including the manipulator, we have $n+1$ voters in total) and $a$ appears within the first $k$ positions in at least $(\lceil\frac{n}{2}\rceil+1)$ positions if $k<t$ or in $\lceil\frac{n}{2}\rceil$ positions, if $k\ge t$. This happens if and only if there exists an integer $z\in\{0,1,\ldots,n_2\}$ ($z$ corresponds to the number preferences of type (ii) which are modified to put $c$ outside the first $k$ positions) such that we have $n_2-z+n_4+n_5\le \lceil\frac{n}{2}\rceil-1$ and, if $k<t$, then $n_1+n_2 -z+n_4\ge\lceil\frac{n}{2}\rceil+1$ else, if $k\ge t$, then $n_1+n_2 - z+n_4\ge\lceil\frac{n}{2}\rceil$. This concludes the description of our algorithm. Our algorithm clearly runs in polynomial-time. \longversion{We next argue its correctness.}\shortversion{In the interest of space, we omit the proof of correctness.}
\longversion{
Suppose that the algorithm outputs that the instance is a \YES instance. Then we claim that the manipulator's preference $\suc_M$ constructed by the algorithm results in a successful manipulation. Suppose not, then there exists an $n$-voters preference profile \QQ such that (i) the Kendall Tau distance between the $i$-th preferences of \PP and \QQ is at most $\delta_i$ and (ii) $c$ is not a simplified Bucklin winner in $(\QQ,\suc_M)$, that is there exists an alternative $a\in\AA\setminus\{c\}$ and a position $k\in[m]$ such that $c$ does not appear within the first $k$ positions in a majority of the preferences whereas $a$ appears within the first $k$ positions in a majority of preferences in $(\QQ,\suc_M)$. Suppose the alternative $a$ appears at the $t$-th position in $\suc_M$. Also let the number of preferences of type (ii) where $c$ is put outside the first $k$ positions in \QQ be z. Then we consider the profile $\QQ^a$ obtained from \PP where
\begin{itemize}
 \item in preferences of type (i), simultaneously $c$ is placed outside the first $k$ positions and $a$ is placed within the first $k$ positions in $\QQ^a$.
 
 \item in z number of preferences of type (ii) and all preference of type (iii), $c$ is put outside the first $k$ positions in $\QQ^a$. In $n_2-z$ number of preferences of type (ii) and all preference of type (iv), $a$ is put within the first $k$ positions in $\QQ^a$.
 
 \item all preferences of type (v) remain the same in \PP and $\QQ^a$.
\end{itemize}
It follows that, since $c$ does not get a majority within the first $k$ positions but $a$ gets a majority within the first $k$ positions in $(\QQ,\suc_M)$, $c$ does not get a majority within the first $k$ positions but $a$ gets a majority within the first $k$ positions in $(\QQ^a,\suc_M)$. This contradicts our assumption that the algorithm declared the position $t$ in the manipulator's preference safe for the alternative $a$. Hence the input instance is indeed a \YES instance.

Now suppose that the algorithm outputs that the instance is a \NO instance. For the sake of arriving to a contradiction, let us assume that there exists a manipulator's preference $\suc_M^\pr\in\LL(\AA)$ which results in successful manipulation. Since our algorithm outputs \NO, there exists an iteration $t\in\{2,\ldots,m\}$ such that, if $\AA_{t-1}$ is the set of alternatives already placed in the first $t-1$ positions by the algorithm, then every alternative $a\in\AA\setminus\AA_{t-1}$ was judged unsafe for the position $t$ in the manipulator's preference. In this case, there indeed exists an $n$-voters profile $\QQ^a\in\LL(\AA)^n$ for every alternative $a\in\AA\setminus\AA_{t-1}$ such that (i) the Kendall Tau distance between the $i$-th preferences of \PP and $\QQ^a$ is at most $\delta_i$ and (ii)  $c$ appears within the first $k$ positions in at most $(\lceil\frac{n}{2}\rceil-1)$ preferences in $\QQ^a$ (observe that, including the manipulator, we have $n+1$ voters in total) and $a$ appears within the first $k$ positions in at least $(\lceil\frac{n}{2}\rceil+1)$ positions if $k<t$ and $\lceil\frac{n}{2}\rceil$ positions if $k\ge t$ . We observe that, if a position $k$ is unsafe for an alternative $x\in\AA\setminus\{c\}$, then the position $k-1$ is also unsafe for $x$. Then we have $m-t+1$ alternatives, namely the alternatives in the set $\AA\setminus\AA_{t-1}$, who must appear within the rightmost $m-t$ positions of any manipulator's preference $\suc_M^\pr$ if $\suc_M^\pr$ results in a successful manipulation which is, by pigeonhole principle, impossible. Hence the input instance was indeed a \NO instance and thus the algorithm is correct.}
\end{proof}
}

The main idea of \Cref{thm:sim_buck} can be extended to design a polynomial-time algorithm for the Bucklin voting rule.

For the Bucklin rule, we classify votes into a constant number of meta-types, the set of all types which can be reached from a given vote. We then show that we can efficiently enumerate all possible ways of transforming preferences into one of the types in their meta-types, noting that only the number of preferences which are converted to a given type is relevant to the rest of the algorithm.  
\begin{theorem}\label{thm:buck}
There exists a polynomial-time algorithm for the \DRSM problem for the Bucklin voting rule if we have only one manipulator.
\end{theorem}

\longversion{
We use the same greedy strategy as in the case of the Simplified Bucklin rule. However, this time, the conditions for safety are more complicated, and are not obviously checkable. We observe that only the numbers of certain suitable types of the votes of others' need to be fixed to check safety, not the votes themselves, hence obtaining a polynomial-time algorithm.

\begin{proof}
Let $(\AA,\PP,c,(\delta_i)_{i\in[n]},\el=1)$ be an arbitrary instance of \DRSM for the Bucklin voting rule. On a high level, our algorithm for the Bucklin voting rule is similar to our algorithms in \Cref{thm:sc,thm:maximin,thm:sim_buck}. The only difference being, given a position $t\in\{2,\ldots,m\}$ and an alternative $a\in\AA\setminus\{c\}$, how do we decide if placing the alternative $a$ at position $t$ in the manipulator's vote is safe. We describe this below and skip repeating the other parts.

Let us denote by $S_k(\RR,a)$ the $k$-approval score of alternative $a$ in an $n$-voters profile $\RR$. We observe that the alternative $c$ wins in \RR under the Bucklin rule if and only if, for every other alternative $a\in\AA\setminus\{c\}$ and for every $k\in \{1,2,\ldots m\}$, $S_k(\RR,a) > \frac{n}{2} \implies S_{k-1}(\RR,c) > \frac{n}{2}$ or $S_k(\RR,c)\ge S_k(\RR,a)$. Given a preference $\suc\in\LL(\AA)$ and an alternative $a\in\AA\setminus\{c\}$, we define the following set $\XX=\{x_1,x_2,x_3,x_4\}$ of Boolean variables.
\begin{enumerate}[label=(\roman*)]
	\item We say that \suc satisfies $x_1$ if and only if the alternative $c$ does not appear within the first $k$ positions in $\suc$.
	
	\item We say that \suc satisfies $x_2$ if and only if the alternative $c$ does not appear within the first $k-1$ positions in $\suc$.
	
	\item We say that \suc satisfies $x_3$ if and only if the alternative $a$ appears within the first $k-1$ positions in $\suc$.
	
	\item We say that \suc satisfies $x_4$ if and only if the alternative $a$ appears within the first $k$ positions in $\suc$.
\end{enumerate}

We define the {\em ``type"} $\TT(\suc,a)\subseteq\XX$ of the preference \suc with respect to an alternative $a\in\AA\setminus\{c\}$ to be the subset of $\XX$ satisfied by the preference \suc. Before we explain our algorithm for deciding whether a position $t$ is safe for an alternative $a\in\AA\setminus\{c\}$, we need to define a few concepts and notations. Given a preference $\suc\in\LL(\AA)$, an alternative $a\in\AA\setminus\{c\}$, and a distance $\delta$, we define the {\em``meta-type"} $\MM(\suc,a,\delta)$ of \suc with respect to $a$ and $\delta$ as the set of all types reachable from \suc within a Kendall Tau distance of at most $\delta$; that is $\MM(\suc,a,\delta) = \{Z\subseteq\XX: \exists \suc^\pr\in\LL(\AA), d_{KT}(\suc,\suc^\pr)\le\delta, \TT(\suc^\pr,a)=Z\}\subseteq 2^\XX$. Let the set of all possible meta-types be $\MM=\{M_i: i\in[\nu]\}$. An important observation is that, since \XX has only $4$ elements, only $2^{2^4}=65536$ (which is a constant) different meta-types are possible. For ease of exposition, let us define $\gamma=16, \nu=65536$. For an alternative $a\in\AA\setminus\{c\}$, let $n_i$ be the number of preferences in \PP of meta-type $M_i$ and $\lambda_i$ the number of types in the meta-type $M_i$; that is $M_i=\{T_{i,1},\ldots,T_{i,\lambda_i}\}$. Another important observation is that, given any preference $\suc\in\LL(\AA)$, a distance $\delta$, and a type $T\in\XX$, it can be checked in polynomial-time whether $T\in\MM(\suc,a,\delta)$; hence the set $\MM(\suc,a,\delta)$ can be computed in polynomial-time.

We now describe our algorithm for whether $a$ is safe at position $t$. For every $k\in[m]$, we check the following. For a tuple $\tau=(\tau_i)_{i\in[\nu]}$ where $\tau_i=(\el_j)_{j\in[\lambda_i]}$ such that $\sum_{j=1}^{\lambda_i} \el_{j} = n_i$, we define an $n$-voters preference profile $\QQ^\tau$ constructed by converting $\el_{j}$ number of preferences of meta-type $M_i$ to preferences of type $T_{i,j}$. Clearly the Kendall Tau distance between the $i$-th preferences of \PP and $\QQ^\tau$ is at most $\delta_i$. We declare position $t$ to be safe for the alternative $a$ if and only if, for every $k\in[m]$ and every possible corresponding $\tau$, the condition $S_k(\RR,a) > \frac{n}{2} \implies S_{k-1}(\RR,c) > \frac{n}{2}$ or $S_k(\RR,c)\ge S_k(\RR,a)$ holds for the profile $\RR = (\QQ^\tau,\suc_M)$ where $\suc_M$ is any manipulator's preference where $c$ and $a$ are placed respectively at positions $1$ and $t$ (Notice that the $k-1$ and $k$ approval scores of alternatives $c$ and $a$ are fixed for a fixed $\tau$).
 This concludes the description of our algorithm. Our algorithm runs in polynomial-time since there are $\OO(n^{\gamma\nu}\text{poly}(m))$ possible tuples $\tau$. We next argue its correctness.

Suppose that the algorithm outputs that the instance is a \YES instance. Then we claim that the manipulator's preference $\suc_M$ constructed by the algorithm results in a successful manipulation. Suppose not, then there exists an $n$-voters preference profile \QQ such that (i) the Kendall Tau distance between the $i$-th preferences of \PP and \QQ is at most $\delta_i$ and (ii) $c$ is not a Bucklin winner in $(\QQ,\suc_M)$, that is there exists an alternative $a\in\AA$ at some position $t$ in $\suc_M$ and a position $k$ such that (i) if $k< t$, then $S_k(\QQ,a) \ge \lceil\frac{n}{2}\rceil+1, S_{k-1}(\QQ,c) \le \lceil\frac{n}{2}\rceil-1$, and $S_k(\QQ,a) > S_k(\QQ,c)$, and (ii) if $k\ge t$, then $S_k(\QQ,a) \ge \lceil\frac{n}{2}\rceil+1, S_{k-1}(\QQ,c) \le \lceil\frac{n}{2}\rceil-1$, and $S_k(\QQ,a) > S_k(\QQ,c)$. Let us consider the set \XX of Boolean variables with respect to the position $k$ and the alternative $a$. Suppose the alternative $a$ appears at the $t$-th position in $\suc_M$. Let us define a tuple $\tau=(\tau_i)_{i\in[\nu]}$ where $\tau_i=(\el_j)_{j\in[\lambda_i]}$ such that $\el_{j}$ is the number of preferences of meta-type $M_i$ that are converted into a preferences of type $T_{i,j}$ in \QQ. Then it follows that (i) if $k< t$, then $S_k(\QQ^\tau,a) \ge \lceil\frac{n}{2}\rceil+1, S_{k-1}(\QQ^\tau,c) \le \lceil\frac{n}{2}\rceil-1$, and $S_k(\QQ^\tau,a) > S_k(\QQ^\tau,c)$, and (ii) if $k\ge t$, then $S_k(\QQ^\tau,a) \ge \lceil\frac{n}{2}\rceil, S_{k-1}(\QQ^\tau,c) \le \lceil\frac{n}{2}\rceil-1$, and $S_k(\QQ^\tau,a) > S_k(\QQ^\tau,c)$ which contradicts the fact that the algorithm declared the position $t$ in the manipulator's preference to be safe for the alternative $a$. Hence the instance is indeed a \YES instance.

Now suppose that the algorithm outputs that the instance is a \NO instance. For the sake of arriving to a contradiction, let us assume that there exists a manipulator's preference $\suc_M^\pr\in\LL(\AA)$ which results in successful manipulation. Since our algorithm outputs \NO, there exists an iteration $t\in\{2,\ldots,m\}$ such that, if $\AA_{t-1}$ is the set of alternatives already placed in the first $t-1$ positions by the algorithm, then every alternative $a\in\AA\setminus\AA_{t-1}$ was judged unsafe for the position $t$ in the manipulator's preference. In this case, there indeed exists an $n$-voters profile $\QQ^a\in\LL(\AA)^n$ for every alternative $a\in\AA\setminus\AA_{t-1}$ such that (i) the Kendall Tau distance between the $i$-th preferences of \PP and $\QQ^a$ is at most $\delta_i$ and (ii) there exists some position $k$ such that (a) if $k< t$, then $S_k(\QQ^a,a) \ge \lceil\frac{n}{2}\rceil+1, S_{k-1}(\QQ^a,c) \le \lceil\frac{n}{2}\rceil-1$, and $S_k(\QQ^a,a) > S_k(\QQ^a,c)$, and (b) if $k\ge t$, then $S_k(\QQ^a,a) \ge \lceil\frac{n}{2}\rceil+1, S_{k-1}(\QQ^a,c) \le \lceil\frac{n}{2}\rceil-1$, and $S_k(\QQ^a,a) > S_k(\QQ^a,c)$. We observe that, if a position $k$ is unsafe for an alternative $x\in\AA\setminus\{c\}$, then the position $k-1$ is also unsafe for $x$. Then we have $m-t+1$ alternatives, namely the alternatives in the set $\AA\setminus\AA_{t-1}$, who must appear within the rightmost $m-t$ positions of any manipulator's preference $\suc_M^\pr$ if $\suc_M^\pr$ results in a successful manipulation which is, by pigeonhole principle, impossible. Hence the input instance was indeed a \NO instance and thus the algorithm is correct.
\end{proof}
}

Due to \Cref{thm:sc,thm:maximin,thm:sim_buck,thm:buck}, one may suspect that there may exist a generic algorithm for the \DRSM problem with one manipulator which works for the class {\em responsive and monotone} voting rules that Bartholdi et al. defined~\cite{bartholdi1989computational}. Our next result refutes such a possibility as we show that the \DRSM problem is \coNPH for the Copeland$^\alpha$ voting rule for every $\alpha\in[0,1]$ which is a responsive and monotone voting rule. We reduce from the \XTCC problem which is the complement of the classical \NPC problem \XTC.
\longversion{The \XTC and \XTCC are defined as follows.

\begin{definition}[\XTC and \XTCC]
	Given an universe $\UU$ of $3n$ elements and a collection \SS of $m$ subsets of \UU each containing $3$ elements, compute if there exists a sub-collection $\WW\subseteq\SS$ such that (i) $|\WW|=n$ and (ii) $\cup_{S\in\WW} S = \UU$. An instance $(\UU,\SS)$ of \XTC is called a \YES instance if there indeed exists such a \WW; otherwise it is called a \NO instance. \XTCC is the complement problem of \XTC: an instance  $(\UU,\SS)$ of \XTCC is a \YES instance if and only if $(\UU,\SS)$ is a \NO instance of \XTC.
\end{definition}

Since \XTC is \NPC, it follows that \XTCC is \coNPH. We use the following lemma in our proof.
\begin{lemma}\label{existence_lemma}
	Let $\AA = \BB\cup \Gamma$ be a set of alternatives, with $|\BB| = l$ and $(Z_{(a,b)}), a,b\in \BB, a\neq b$ be integers, all with the same parity, satisfying $Z_{(b,a)} = - Z_{(a,b)} \;\forall a,b\in \BB $. Let $\delta \geq 0$, a positive integer, be given. Further, suppose that $|\Gamma| \geq 10\delta \el\sum_{a,b\in \BB, a\neq b} |Z_{(a,b)}|. $ Then there exists a preference profile $\PP = (\suc_i)_{i \in [|\PP|]}$ on the set of alternatives $\AA$ satisfying
	\begin{enumerate}
		\item $\DD_\PP(a,b) = Z_{(a,b)} \;\forall\; a,b\in \BB, a\neq b $
		\item For any two alternatives $a,b \in \BB, a \neq b$ and any preference $\suc \in \PP$, $|\text{rank}(\suc,a) - \text{rank}(\suc,b)| > \delta $.
		\item For any two alternatives $b \in \BB, d\in \Gamma$, $d$ appears before $b$ in at most one preference in $\PP^\pr$  for all profiles $\PP^\pr = (\suc_i^\pr)_{i \in [|\PP|]} $  which satisfy $d_{KT}(\suc_i,\suc_i^\pr) \leq \delta\; \forall\; i \in [|\PP|]$
		\item The number of preferences in $\PP$ is bounded by a polynomial function of $\sum_{a,b\in B, a\neq b} |Z_{(a,b)}|$, and $\PP$ can be constructed in time polynomial in $m+\sum_{a,b\in B, a\neq b} |Z_{(a,b)}|$
	\end{enumerate}
\end{lemma}

\begin{proof}
	Follows from the proof of Lemma 13 in \cite{Dey19}.
\end{proof}

The set of alternatives $\Gamma$, described above in \Cref{existence_lemma} can be thought of as dummy alternatives - they help us to `control' the results of pairwise elections of the alternatives in $\BB$, essentially fixing most outcomes of the pairwise elections. We then suitably specify other votes to ensure that only the outcomes of the pairwise elections among `the alternatives of interest' change suitably to ensure the correctness of the reduction. We now prove the hardness result.}

The idea is to use an alternative $x$, who will be the sole contender for the distinguished alternative $c$. We then add alternative $y_u$ corresponding to each element of the set $U$ in the \XTCC instance. We then ensure that $c$ co-wins only if it defeats each $y_u$, and $x$ is defeated by another specified alternative $z$ (this ensures that the corresponding exact cover obtained is of size $\frac{|U|}{3}$).

\begin{theorem}\label{thm:copeland}
The \DRSM problem is \coNPH for the Copeland$^{\alpha}$ voting rule for every $\alpha \in [0,1]$ even if we have only one manipulator and $\delta=3$ for every preference.
\end{theorem}

\longversion{
	
Before proving the theorem, we provide an intuitive overview of the proof. The idea is to use an alternative $x$, who will be the sole contender for the distinguished alternative $c$. We then add alternative $y_u$ corresponding to each element of the set $U$ in the \XTCC instance. We then ensure that $c$ co-wins only if it defeats each $y_u$, and $x$ is defeated by another specified alternative $z$. To ensure that $c$ and $x$ are the only possible Copeland winners, and also to ensure other unintended swaps do not change any outcome, we add dummy alternatives.

In the course of the proof, we use two sets of vote profiles, $\PP_1$ and $\PP_2$, with $\PP = (\PP_1,\PP_2)$. Each preference of $\PP_1$ is explicitly constructed. The preferences in $\PP_2$ are not explicitly specified. Instead, we specify $D_\PP(x,y)$ for each pair of alternatives $x,y, x\ne y$. Notice that $D_\PP(x,y) = D_{\PP_1}(x,y) + D_{\PP_2}(x,y)$, fixing $D_\PP(x,y)$ thus fixes $D_{\PP_2}(x,y)$. We then invoke \Cref{existence_lemma} to guarantee the existence of such a profile $\PP_2$.

\begin{proof}
We now prove co-NP hardness. We reduce from \XTCC. Let $(\UU = \{u_i: i\in[3n]\}, \SS=\{S_j:j\in[m]\})$ be an arbitrary instance of \XTCC. We consider the following instance $(\AA,\PP=(\PP_1,\PP_2),c,(\delta_i=3)_{i\in[n]},\el=1)$ of \DRSM. 
\begin{align*}
	\AA &= \BB \cup \Lambda\cup \Gamma  \text{ where }\\
	\BB &= \{c,x,z\}\cup\{y_u: u\in\UU\}, |\Lambda|=100n, |\Gamma| = K m^3n^3\\
	\PP_1 &= \{c\suc \{y_u: u\in S\}\suc d_1\suc d_2\suc d_3\suc\\
	&z\suc d_4\suc d_5\suc x\suc\text{others}\\ &\text{for some } d_i\in\Gamma, i\in[5]: \forall S\in\SS\}
\end{align*}

The constant $K$ is chosen suitably to satisfy the conditions of the Lemma.
While adding preferences in $\PP_1$ we use a new set of dummy candidates $d_1$ through $d_5$ for every vote.
Also, in the preferences above, whenever we say `others', the unspecified alternatives
are assumed to be arranged in such a way that, for every unspecified alternative $a\in\AA\setminus\Gamma$, the three alternatives to both the immediate left and right of $a$ are all from $\Gamma$. 
 We now add a preference profile $\PP_2$ such that we have the following.

\begin{itemize}
	\item  $\DD_\PP(c,y_u)=\DD_\PP(x,y_u)=0$ for every $u\in\UU$
	
	\item $\DD_\PP(z,c)= 0$
	
	\item $\DD_\PP(y_u,z) = 8m$ for every $u\in\UU$
	
	\item $\DD_\PP(z,x)=2(m-n)-2$

	\item $\DD_\PP(x,c) = 8m$
\end{itemize}

In \PP, we further ensure the following. Every alternative in $\Lambda$ gets defeated by at least $\frac{1}{3}|\Lambda|$ of the alternatives from $\Lambda$ in pairwise elections by a margin of $8m$. Every alternative in $\{c,x\}$ defeats every alternative in $\Lambda$ in pairwise elections by a margin of $8m$. Every alternative in $\Lambda$ defeats every alternative in $\{y_u:u\in\UU\}\cup\{z\}$ in pairwise elections by a margin of $8m$. 
Such a profile $\PP_2$ (and thus \PP) exists due to \Cref{existence_lemma} (applied as $\BB=\AA\setminus\Gamma$), for a suitable choice of $K$ (We can always choose such a $K$, since $D_{\PP_2}(x,y)$ is $\OO(m + n)$, for each pair of alternatives $(x,y) \in \AA \setminus \Gamma$). Further, \Cref{existence_lemma} guarantees that for every alternative $a$ in $\AA\setminus\Gamma$, the 3 alternatives to both the immediate left and right of $a$ are all from $\Gamma$, for each preference in $\PP_2$. Property (3) in \Cref{existence_lemma} ensures that no alternative from  $\Gamma$  wins for any choice of manipulator vote. Thus, it suffices to consider only alternatives in $\AA \setminus \Gamma$ as candidate Copeland winners. We now claim that the two instances are equivalent.

In one direction, suppose the \XTCC instance is a \YES instance. Then there does not exists an exact cover $\WW\subseteq\SS$ for \UU, with $|\WW| = n$. We claim that the manipulator's vote $\suc_M=c\suc\text{others}\suc x$ results in a successful manipulation. To see this, let \QQ be any $n$-voters profile such that the Kendall Tau distance between the $i$-th preferences of \PP and \QQ is at most $3$. We first observe that, since every alternative in $\Lambda\cup\{y_u:u\in\UU\}\cup\{z\}$ gets defeated by at least $\frac{1}{3}|\Lambda|$ of the alternatives from $\Lambda$ and, for every unspecified alternative $a\in\AA\setminus\Gamma$, the 3 alternatives to both the immediate left and right of $a$ are all from $\Gamma$ and $\delta=3$, no alternative in $\AA\setminus\{c,x\}$ wins in \QQ. Let us define $\WW$ to be the set of $S\in\SS$ such that in the corresponding preference in \QQ, $c$ does not appear at the first position.  Notice that $c$ defeats $z$ in $(\QQ,\suc_M)$. If $|\WW|> n$, then the alternative $z$ defeats $x$ in $(\QQ,\suc_M)$, and consequently  $c$ is a co-winner in $(\QQ,\suc_M)$. On the other hand, if $|\WW|\leq n$, then \WW is not an exact set cover for \UU, there exists an element $u\in\UU$ such that $c$ defeats the alternative $y_u$ in $(\QQ,\suc_M)$ and thus $c$ is a co-winner in $(\QQ,\suc_M)$.

For the other direction, suppose the \XTCC instance is a \NO instance. Let $\WW\subset\SS$ form an exact set cover for \UU. Let us consider the $n$-voters preference profile obtained from \PP as: for every $S\in\WW$, in the corresponding preference, we shift $c$ to right by $3$ positions; for every $S\in\SS\setminus\WW$, in the corresponding preference, we shift $x$ to left by $3$ positions. It follows that, irrespective of the manipulator's preference $\suc_M\in\LL(\AA)$, every alternative in $\{y_u: u\in\UU\}$ defeats $c$ and $x$ defeats $z$. Hence $x$ defeats $1$ more alternative than $c$ and thus the \DRSM instance is a \NO instance.
\end{proof}
}

\subsection{Results for \DRWM}

\longversion{
We make use of a version of a related problem from \cite{Dey19} while proving our results in this section:
\begin{definition}
\LB  Given a set $\AA$ of alternatives, a profile $\PP = (\suc_i{i \in [n]})$ consisting of $n$ voters,  a positive integer $\delta$, compute if there exists a vote profile $\QQ = (\suc_i^\pr{i \in [n]})$ where $d_{KT}(\suc_i,\suc_i^\pr) \leq \delta$, for each $i \in [n]$, so that $c$ is a co-winner in $\QQ$.   

\end{definition}
}
We reduce the problem for plurality rule to a bribery problem in \cite{Dey19}, which uses a maximum flow reduction.
\begin{theorem} \label{thm:plurality_pm}
\DRWM is polynomial-time solvable for the plurality rule.
\end{theorem}
\longversion{

\begin{proof}
Let $(\AA,\PP=(\suc_i)_{i\in[n]},c,(\delta_i)_{i\in[n]},\el)$ be an arbitrary instance of \DRWM. Pick an arbitrary manipulators' profile $\suc_M$ that places $c$ in the first place in each vote, and let $\QQ = \PP \cup \{\suc_M\}$. We now use the algorithm for the plurality rule for \LB for $\QQ$, with $\delta = 0$ for each manipulator's vote added above. Since it is known that \LB is poly-time solvable for the plurality rule even when each vote $i$ has a different distance parameter $\delta_i$, \DRWM is polynomial-time solvable for the plurality rule. 

\end{proof}
}

Clearly, the \DRWM problem is in \NP for all common voting rules. Thus, in all our hardness results we prove only \NP-hardness.

\begin{observation}
If \MAN is \NPH for a voting rule $r$ with $l$ manipulators, then \DRWM is \NPH for $r$ with a single manipulator. 
\end{observation}
\begin{proof}
Let $(\AA,\PP=(\suc_i)_{i\in[n]},c,\el)$ be an arbitrary instance of \MAN. We reduce to an instance of \DRWM as follows: We add $l-1$ votes with $\delta = {m \choose 2}$ to $\PP$ to obtain a new profile $\QQ$. Then the reduced instance of \DRWM is $(\AA,\QQ,c,(\delta_i)_{i \in [n + l - 1]},1)$, where $\delta_i = 0$ for $i \in [n]$, and $\delta_i = {m \choose 2}$, for $i \in \{n+1, n+2 \ldots n+l-1\}$. It is immediate that the two instances are equivalent.

\end{proof}

For the $k$-approval rule, we reduce from an unbudgeted bribery problem in \cite{Dey19}. We add sufficiently many dummy candidates, and add further votes so that in the reduced instance, the score of every alternative including the manipulator's vote increases by exactly $1$, so that the instances are equivalent.
\begin{theorem}\label{thm:kapp_pm}
\DRWM is \NPC for the k-approval rule, for even a single manipulator, for any constant $k \geq 2$, even with $\delta = 2$ for each vote.
\end{theorem}
\longversion{
\begin{proof}
We prove hardness by reduction from \LB, which is known to be \NPC even when $\delta = 2$ for each vote.
Let $(\AA,\PP=(\suc_i)_{i\in[n]},c,(2)_{i\in[n]})$ be instance of \LB. Let $\Gamma = \{d_{ai}\; |\;a \in \AA , i \in [k+1]\}$ and $\BB = \AA \cup \Gamma $. We construct the vote profile $\QQ$, by appending to the end of each preference in $\PP$, the candidates from $\Gamma$ in arbitrary order. We then further add the following preferences in $Q$ - for each $a \in \AA\setminus\{c\}$, we add an arbitrary completion of the preference, $a \suc d_{a1} \suc d_{a2} \ldots \suc d_{a(k+1)}$. The instance of \DRWM is then $(\BB,\QQ,c,(\delta_i = 2)_{i\in[n + m - 1]},1)$. To prove correctness, we notice that the candidates $\{d_{ci}\; |\; i \in [k+1]\}$, receive zero score in $\QQ$, irrespective of any allowed swaps. Thus we may assume without loss of generality that the manipulator's vote is $c \suc \{d_{ci}\; |\; i \in [k+1]\} \suc others$. We observe that no candidate in $\Gamma$ gets a $k$-approval score of more than 1. Since the added votes together with the manipulator's vote increase the score of every candidate from $\AA$ by exactly 1(even after swapping alternatives within the distance limit), it follows that the two instances are equivalent.

\end{proof}
}

For the Copeland rule, we use a reduction from \XTC similar to that in \Cref{thm:copeland}. 
\begin{theorem}\label{thm:copeland_pm}
\DRWM is \NPC for the Copeland$^{\alpha}$ rule for every $\alpha \in [0,1]$ even with $\delta = 3$ for each vote, for a single manipulator.
\end{theorem}
\longversion{
\begin{proof}
The proof of this result involves a construction very similar to that presented in \Cref{thm:copeland}. We only provide the changes in the construction and omit the (very similar) proof of correctness. This time, however, we shall reduce from \XTC, instead of its complement.
We now let $\PP_1 = \{\{y_u: u\in S\} \suc c \suc d_1\suc d_2\suc d_3\suc
	x\suc d_4\suc d_5\suc z\suc\text{others} ,\text{for some } d_i\in \Gamma, i\in[5]: \forall S\in\SS\}$.
We now ensure the following weights in the majority graph of the overall profile $\PP$:

\begin{itemize}
	\item  $\DD_\PP(c,y_u)= -2,\DD_\PP(x,y_u) = 8m$ for every $u\in\UU$
	
	\item $\DD_\PP(z,c) = 0$
	
	\item $\DD_\PP(y_u,z) = 8m$ for every $u\in\UU$
	
	\item $\DD_\PP(x,z) = 2(m-n)$

	\item $\DD_\PP(x,c) = 8m$
\end{itemize}

The rest of the construction remains the same as in \Cref{thm:copeland}. It follows that $c$ co-wins in the \DRWM instance if and only if there is an exact cover of size $n$.

\end{proof}
}

The result for Maximin rule uses a reduction from \XTC. This time, we add two alternatives $w$, $x$, as potential competitors for the alternative $c$. The alternative $w$ ensures that $c$ can win only if it is swapped ahead of each alternative corresponding to the universe of the exact cover instance atleast once, and the alternative $x$ ensures that this must happen by changing only $\frac{|U|}{3}$ votes, giving an exact cover of size $\frac{|U|}{3}$.
\begin{theorem}\label{thm:maximin_pm}
\DRWM is \NPC for the Maximin voting rule, even with $\delta = 3$ for every vote, for a single manipulator.
\end{theorem}

\longversion{
\begin{proof}
We reduce from \XTC. Let $(\UU = \{u_i: i\in[3n]\}, \SS=\{S_j:j\in[m]\})$ be an arbitrary instance of \XTC. We consider the following instance $(\AA,\PP=(\PP_1,\PP_2),c,(\delta_i=3)_{i\in[n]},\el=1)$ of \DRWM. 
\begin{align*}
	\AA &= \BB \cup \Lambda\cup \Gamma  \text{ where }\\
	\BB &= \{c,x,w,z\}\cup\{y_u: u\in\UU\}, |\Lambda|=100n, |\Gamma| = K m^3n^3\\
	\PP_1 &= \{\{y_u: u\in S\} \suc c \suc d_1\suc d_2\suc d_3\suc\\
	&x\suc d_4\suc d_5\suc z\suc d_6 \suc d_7 \suc d_8 \suc w \suc \text{others}\\ &\text{for some } d_i\in\Gamma, i\in[8]: \forall S\in\SS\}
\end{align*}
As in \Cref{thm:copeland}, we create a vote profile $\PP_2$ satisfying the conditions of \Cref{existence_lemma}, so that the combined vote profile $\PP = \PP_1 \cup \PP_2$, satisfies the folowing:

\begin{itemize}
	\item  $\DD_\PP(c,y_u) = -6, \DD_\PP(x,y_u) = 8m, \DD_\PP(w,y_u) = 8m $ for every $u\in\UU$
	
	\item $\DD_\PP(z,c)= 0, \DD_\PP(w,c) = -2, \DD_\PP(x,c) = 0$
	
	\item $\DD_\PP(y_u,z) = \DD_\PP(w,z) = 8m, \DD_\PP(c,z) = -4$ for every $u\in\UU$
	
	\item $\DD_\PP(x,z) = 2(m-n) - 2, \DD_\PP(w,x) = 0$ 

\end{itemize}

In \PP, we further ensure the following. Every alternative in $\Lambda$ gets defeated by at least $\frac{1}{3}|\Lambda|$ of the alternatives from $\Lambda$ in pairwise elections by a margin of $8m$. Every alternative in $\{c,x,w\}$ defeats every alternative in $\Lambda$ in pairwise elections by a margin of $8m$. Every alternative in $\Lambda$ defeats every alternative in $\{y_u:u\in\UU\}\cup\{z\}$ in pairwise elections by a margin of $8m$. We ensure that no candidate from $\Gamma$ wins for any choice of vote profile within the distance requirements, and any choice of manipulator vote.

By construction, we observe that no candidate in $\AA \setminus \{c,x,w\}$ wins in $(\QQ,\suc_M)$, for any profile $\QQ$, which meets the distance restriction with respect to the profile $\PP$.

We now show that the two instances are equivalent. Suppose that the \DRWM instance is a \YES-instance. Let $\QQ$ be any profile that meets the distance criterion with respect to $\PP$ and $\suc_M$ be a successful manipulator vote. We may assume without loss of generality that $\suc_M$ places $c$ at the first position. Now we observe that the candidate $w$ has a maximin score of at least $-3$ in $(\QQ,\suc_M)$. Since the manipulator's vote has $c$ at the first position, in order for $c$ to co-win in $(\QQ, \suc_M)$ it follows that we must have $D_{\QQ}(c,y_u) \geq -4$, for each $u \in U$. Let $W$ be the set of $S \in \SS$, so that in the corresponding preference in $\QQ$, $c$ appears within the first three positions(i.e. $c$ is shifted by atleast one place to the left with respect to the corresponding vote in $\PP_1$). Since $D_{\PP}(c,y_u) = -6$, $c$ must be swapped ahead of each $y_u$ atleast once in $\QQ$. It follows that $W$ must constitute an exact cover $W$ of $U$. It remains to show that $|W| = n$. Assume to the contrary that $|W| > n$. We then have that the candidate $x$ loses against $z$ by a margin of at most $1$ in $(\QQ,\suc_M)$. However, in that case, the maximin score of $x$ in $(\QQ,\suc_M)$ is at least $-1$, but the maximin score of $c$ in $(\QQ,\suc_M)$ is at most $-3$ (since $D_Q(c,z) = -4$), which contradicts the fact that $c$ co-wins in $(\QQ,\suc_M)$ . Thus the $X3C$ instance is a \YES-instance.

Now suppose conversely that the X3C instance is a \YES instance, and let $W$ be an exact cover of size $n$. We move $c$ three positions to the left in each vote corresponding to the sets in $W$ and shift $z$ to the left by three positions in the remaining $m - n$ votes in $\PP_1$. Let us call the obtained profile $\QQ$. Consider the manipulator vote $\suc_M = c \suc others \suc w \suc x$. Observe that $D_{(\QQ, \suc_M)}(x, z) = -3$. The maximin score of $w,x,c$ are all $-3$, and thus $c$ co-wins in the profile $(\QQ,\suc_M)$.

\end{proof}
}

We now prove our result for the Borda rule. We make use of the \PS problem, which is known to be \NPC~\cite{yu2004minimizing}. We remark that a reduction from \PS was first used to prove that Borda manipulation is \NPC for two or more manipulators~\cite{davies2014complexity}.

\begin{definition}{(\PS)}
Given a set of $l$ positive integers $x_1,x_2 \ldots x_l$, with $\sum_{i \in [l]}{x_i} = l(l+1)$, do there exist two permutations $\sigma: [l] \rightarrow [l]$ and $\pi:[l] \rightarrow [l]$, so that $x_i = \sigma(i) + \pi(i)$, for each $i \in [l]$ ?

\end{definition}
For ease of presentation, in this section, we define the Borda score of a candidate as one plus the number of candidates it ranks above.

\begin{theorem}\label{thm:borda_pm}
\DRWM is \NPC for the Borda rule even with $\delta = 1$ for each vote.
\end{theorem}

\begin{proof}
We reduce from the \PS Problem. Let $(x_1, x_2, x_3 \ldots x_l)$ be an instance of the \PS Problem. We construct an instance of \DRWM $(\AA,\PP = \PP_1 \cup \PP_2,c,(\delta_i = 1)_{i \in [|\PP|]},1)$ , with the set of alternatives $\AA = \CC \cup \{x,y\}$, where $\CC = \{a_i \;|\; i \in [l]\} \cup \\ \{b_i \;| \;i \in [l]\}$. We construct the vote profile $\PP_1$ with the following preferences:
\begin{itemize}
\item $\frac{l(l+1)}{2}$ copies of $c \suc b_1 \suc a_1 \suc b_2 \suc a_2 \ldots b_l \suc a_l \suc x \suc y$
\end{itemize} 
Let us denote the  Borda score of any alternative $t$ in $\PP_1$ as $s_1(t)$. We now construct the profile $\PP_2$ with the following votes for each $i \in [l]$: (\textit{others} refers to all the unspecified alternatives in the preference order):

\begin{itemize}
	\item $s_1(c) + 2l + 2 - s_1(b_i)$ copies of $\textit{others} \suc b_i \suc x \suc y \suc c $
	\item $s_1(c) + 2l + 2 - s_1(b_i)$ copies of $  \ x \suc b_i \suc c \suc y \suc \textit{others}$, where the ordering of \textit{others} is reversed in the second set with respect to the first set.

	\item $s_1(c) + l + 2 - x_i - s_1(a_i)$ copies of $\textit{others} \suc a_i \suc x \suc y \suc c$
	\item $s_1(c) + l + 2 - x_i - s_1(a_i)$ copies of $ x \suc a_i \suc c \suc y \suc \textit{others}$, where the ordering of \textit{others} is reversed in the second set with respect to the first set.
\end{itemize}
To ensure that $x,y$ do not win irrespective of the manipulator vote we add sufficient copies of the following votes in $\PP_2$ (only polynomially many votes need to be added):
\begin{itemize}
	\item $a  \suc b \suc c \suc d \suc \textit{others} \suc x \suc y$ 
	\item $\textit{others} \suc d \suc b \suc a \suc c  \suc x \suc y$, for some $a, b, d \in \CC \setminus \{c\}$, where the ordering of \textit{others} is reversed in the second set with respect to the first.
\end{itemize}
We observe the following:
\begin{itemize}
	\item Since $x$ and $y$ stand no chance of winning, we may assume without loss of generality that the manipulator's vote is $c \suc \{x,y\} \suc \CC$, for some ordering of the set $\CC$.
	\item Let $s_{\QQ}(a)$ denote the score of alternative $a$ in $\QQ$, where $\QQ$ is obtained from $\PP$ by moving $c$ to the left by one position in all the votes of $\PP_2$.\\ Then $s_{\QQ}(b_i) = s_{\QQ}(c) + 2l + 2$, and $s_{\QQ}(a_i) = s_{\QQ}(c) + l + 2 - x_i$.  
\end{itemize}
 
In one direction, suppose the \PS instance is a \YES instance. Let $\sigma$ and $\pi$ be the permutations satisfying the input instance. In the profile $\PP_1$ we shift $b_i$ to the right of $a_i$ by one position in  $\sigma(i)$ many votes, for each $i \in [l]$ (notice that this requires using all the votes in $\PP_1$). In each vote of $\PP_2$ we shift $c$ to the left by one position. We then construct the manipulator's vote $\suc_M$ as $c \suc \{x,y\} \suc a_{\pi^{-1}(l)} \suc a_{\pi^{-1}({l-1})} \suc \ldots a_{\pi^{-1}(1)} \suc b_{\sigma^{-1}(l)} \suc b_{\sigma^{-1}(l-1)} \suc \ldots \suc b_{\sigma^{-1}(1)}$. Let $s(a)$ denote the total Borda score of alternative $a$ in the modified profile, including the manipulator's vote. We then have $s(b_i) = s(c)$, $s(a_i) = s(c)$ for each $i \in [l]$, and $c$ co-wins in the \DRWM instance.

Conversely, suppose the \DRWM instance is a \YES instance, and $\suc_M$ be a successful manipulation of the form $c \suc \{x,y\} \suc \CC$. We assume without loss of generality that $c$ is moved to the left by one position in each preference of $\PP_2$. Every candidate $b_i$ receives a score of $j_i$ from $\suc_M$, where $j_i$ is its position from the right in $\suc_M$. The construction then requires that for every $i \in [l]$, $b_i$'s score must reduce by at least $j_i$ in the votes corresponding to $\PP_1$ to ensure that $c$ co-wins. However, the total number of votes in $\PP_1$ is $\frac{l(l+1)}{2} \leq \sum_{i \in [l]} j_i$, and thus it follows that each $b_i$ must appear in the last $l$ positions from the right in the manipulator's vote. Furthermore, each alternative $b_i$ must have moved right by one position in exactly $j_i$ preferences of the profile $\PP_1$. Thus the score of each candidate $a_i$ is enhanced by $j_i$. Now suppose that the position of the candidate $a_i$ in $\suc_M$ is $k_i + l$ from the right for some $k_i  \in [l]$. Since $c$ wins in the \DRWM instance for the given manipulator vote, we must have $k_i + j_i \leq x_i$. However, $\sum_{i \in [l]}{x_i} = l(l+1)$ and $\sum_{i \in [l]}{j_i} = \frac{l(l+1)}{2}$, $\sum_{i \in [l]} k_i = \frac{l(l+1)}{2}$ which gives $k_i + j_i =  x_i$, for each $i \in [l]$. It follows that the permutations $\sigma$ and $\pi$, given by $\sigma(i) = j_i$, and $\pi(i) = k_i$, satisfy the \PS instance.
\end{proof}
For the simplified Bucklin and Bucklin rules, we modify the reductions from $(3,B2)$-SAT in Theorems 16 and 18 of \cite{Dey19}. We add a single vote that ensures, along with the manipulator's vote, the score of each alternative increases by exactly $1$ in the first few positions. Since the number of votes needed to obtain a majority is now one more than in the original instance, the rest of the analysis remains the same. 
\begin{theorem}\label{thm:sb_pm}
The \DRWM problem is \NPC for the Simplified Bucklin rule, even with $\delta = 1$, for each vote, for a single manipulator.

\end{theorem}

\longversion{
\begin{proof}
The proof of this result is similar to that of Theorem 16 in \cite{Dey19} to which we refer. Consider the reduced \LB instance. Let $\BB = \AA \setminus \DD$. We add the following vote $\suc_s$ to the constructed vote profile $\PP_1$:
\begin{itemize}
	\item $\BB \setminus \{c\} \suc \text{others} \suc c$
\end{itemize}

Since, by construction, we ensure no candidate from $\DD$ wins, we may assume without loss of generality that the manipulator's vote $\suc_M$ is of the form $c \suc \DD \suc \text{others}$. Since $|\AA \setminus \DD| = 5n + m + 2$, we notice that due to the votes $\suc_s$, and $\suc_M$, each candidate in $\AA$ appears exactly once more in the top $mn$ positions. Furthermore, since the total number of votes(including the manipulator's vote) has increased by $2$, the number of votes needed to obtain a majority has increased by exactly $1$. This finishes the proof since the rest of the analysis remains identical to that in \cite{Dey19}.        
\end{proof}
}

\begin{theorem}\label{thm:bucklin_pm}
The \DRWM problem is \NPC for the Bucklin rule, even with $\delta = 1$, for each vote, for a single manipulator.

\end{theorem}

\longversion{
\begin{proof}
The proof of this result is similar to Theorem 18 of \cite{Dey19} to which we refer. Consider the reduced \LB instance. Let $\BB = \AA \setminus \DD$. We then add the following vote $\suc_s$ to the constructed vote profile $\PP_1$	:
\begin{itemize}
	\item $\BB \setminus \{c\} \suc \text{others} \suc c$
\end{itemize}

Since, by construction, we ensure no candidate from $\DD$ wins, we may assume without loss of generality that the manipulator's vote $\suc_M$ is of the form $c \suc \DD \suc \text{others}$. Since $|\AA \setminus \DD| = 3n + 2m + 1$, we notice that due to the votes $\suc_s$, and $\suc_M$, each candidate in $\AA$ appears exactly once more in the top $10(m+n)$ positions. Furthermore, since the total number of votes(including the manipulator's vote) has increased by $2$, the number of votes needed to obtain a majority has increased by exactly $1$. This finishes the proof since the rest of the analysis remains identical to that in \cite{Dey19}.        
\end{proof}
}

We now present a general result.

\begin{theorem}\label{thm:poly}
Both \DRSM and \DRWM are poly-time solvable for any poly-time computable anonymous voting rule if the number of alternatives is $\OO(1)$.
\end{theorem}

\longversion{

\begin{proof}
 Let $(\AA,\PP=(\suc_i)_{i\in[n]},c,(\delta_i)_{i\in[n]},\el)$ be an arbitrary instance of \DRSM. Since the voting rule is anonymous, any preference profile can equivalently be described by the number $n_\suc$ of times a preference $\suc\in\LL(\AA)$ appears in the profile. Hence the number of different anonymous preference profiles is $\binom{n + m! -1}{m!-1} = O((n+m!-1)^{m!-1}) = O(n^{\OO(1)})$. We check, for every possible manipulators' anonymous preference profile $\suc_M$ (there are only $\binom{\el + m! -1}{m!-1} = O((\el+m!-1)^{m!-1}) = O(\el^{\OO(1)})$ such preference profiles), if there exists an $n$-voters anonymous preference profile \QQ (by iterating over all possible anonymous preference profiles) where (i) $c$ does not win in $(\QQ,\suc_M)$ and (ii) the input preference profile \PP can be modified into the anonymous preference profile \QQ. We reduce the problem in (ii) above into a maximum flow problem as follows. In the bipartite graph $\GG=(\VV,\EE,s,t,c:\EE\longrightarrow\NB_{\ge1})$ with
 \begin{align*}
 	\VV &= \VV_1 \cup \VV_2, \text{ where }\\
 	\VV_1 &= \{u_1,u_2,\ldots,u_{n}\}\\
 	\VV_2 &= \{v_{\suc}:\suc\in\LL(\AA)\}\\
 	\EE &= \{(u_i, v_\suc) : i\in [n], \suc\in\LL(\AA),d_{KT}(\suc_i,\suc) \leq \delta_i\}\\
 	&\cup \{(s,u_i): i\in[n]\} \cup \{(v_\suc,t): \suc\in\LL(\AA)\}
 \end{align*}
 The capacity of every edge from $\VV_1$ to $\VV_2$ and from $s$ to $\VV_1$ is $1$. If a preference $\suc\in\LL(\AA)$ appears in \QQ $n_\suc$ number of times, we define the capacity of the edge $(v_\suc,t)$ to be $n_\suc$. It follows that \QQ satisfies the second condition if and only if there is a flow of value $n$ in the above flow network.
 To prove the result for \DRWM, we simply modify condition (i) in above to check that $c$ co-wins in $(\QQ,\suc_M)$.
\end{proof}
 }


\section{Discussion and Future Work}
We have proposed a new model of incorporating manipulators' uncertainty about the non-manipulators' preferences through distance restrictions. Intuitively, in our model, the manipulators' uncertainty is distributed over the whole preference of other voters unlike many existing models in literature. We have studied the complexity of manipulation under two settings - one where we demand a robust manipulating vote, and another where we investigate the possibility of manipulation.

We remark that the information we consider available to the manipulator may be easy for the manipulator to obtain in certain scenarios. For example, consider an election which happens say every year to elect a president of an organization. A powerful person in this organization may have access to the votes of every voter in all past elections. From this data, he/she can compute, for every voter $i$, the average change $c_i$ in the preference of voter $i$ across consecutive elections (measured as the average Kendall-Tau distance between consecutive election votes). Now, last year's preferences can act as a proxy for each voter's believed preference, so that voter $i$'s actual preference is within a ball of approximately $\delta_i = c_i$ from his last year's preference.  

The results obtained show that \DRSM is poly-time solvable for most voting rules for a single manipulator. However, even though this is the case, one must take into account the fact that with sufficiently large uncertainty in the non-manipulators' votes, a large number of input profiles maybe \NO-instances, thus foiling manipulation. The hardness of \DRWM for most voting rules also implies that if we consider a scenario in which the manipulators need to succeed in a threshold fraction of the profiles that are within the given distance limit, that problem would also be hard for most voting rules. We note that the complexity of \DRSM for more than one manipulator for the bucklin and simplified bucklin rules is open as yet.

 As future work, one may examine a variant of the problem in which we require the manipulators to succeed in at least a threshold fraction of the profiles within the distance requirements, and investigate its complexity for voting rules like plurality - for which both weak and strong manipulation problems are poly-time solvable, under distance restrictions. Yet another direction would be to examine the fixed parameter tractability of the considered problems, with respect to various election parameters.

\longversion{In the case of scoring rules, we observe that hardness of manipulation for more than one manipulator implies hardness of \DRSM as well. It would be interesting to investigate if \DRSM remains polynomial-time solvable for the case of more than one manipulator, for scoring rules for which the manipulation problem is poly-time solvable for more than one manipulator(k-approval is one such example for which this is the case, as shown).}

\bibliography{references}


\longversion{\bibliographystyle{alpha}}

\end{document}